\documentclass[11pt]{llncs}

\usepackage{amsmath}
\usepackage{amssymb}
\usepackage{calc}
\usepackage{enumerate}
\usepackage{graphicx, subfigure}
\usepackage{color}
\usepackage{booktabs}
\usepackage{arydshln}

\begin{document}

\title{A Minimum-Labeling Approach for Reconstructing Protein Networks across Multiple Conditions}

\author{Arnon Mazza\inst{1} \and Irit Gat-Viks\inst{2} \and 
Hesso Farhan\inst{3} \and Roded Sharan\inst{1}}
\institute{
Blavatnik School of Computer Science, \newline
Tel Aviv University, Tel Aviv 69978, Israel. \newline
Email: \email{roded@post.tau.ac.il}.
\and
Dept. of Cell Research and Immunology,\newline
Tel Aviv University, Tel Aviv 69978, Israel. 
\and
Biotechnology Institute Thurgau, University of Konstanz,\newline
Unterseestrasse 47, CH-8280 Kreuzlingen, Switzerland.
}

\maketitle

\begin{abstract}
The sheer amounts of biological data that are generated in recent years have driven the development of network analysis tools to facilitate the interpretation and representation of these data. A fundamental challenge in this domain is the reconstruction of a protein-protein subnetwork that underlies a process of interest from a genome-wide screen of associated genes.
Despite intense work in this area, current algorithmic approaches are largely limited to analyzing a single screen and are, thus, unable to account for information on condition-specific genes, or reveal the dynamics (over time or condition) of the process in question. Here we propose a novel formulation for network reconstruction from multiple-condition data and devise an efficient integer program solution for it.
We apply our algorithm to analyze the response to influenza infection in humans over time as well as to analyze a pair of ER export related
screens in humans. By comparing to an extant, single-condition tool we demonstrate the power of our new approach in integrating data from multiple conditions in a compact and coherent manner, capturing the dynamics of the underlying processes.
\end{abstract}

\section{Introduction}
With the increasing availability of high-throughput data, network biology
has become the method of choice for filtering, interpreting and
representing these data. A fundamental problem in network biology is
the reconstruction of a subnetwork that underlies a process of interest
by efficiently connecting a set of implicated proteins (derived by
some genome-wide screen) in a network of physical interactions.
In recent years, several algorithms have been suggested for different variants
of this problem, including the Steiner tree based methods
of~\cite{Beisser10,Huang09}, the flow based approach of~\cite{Lotem09}
and the anchored reconstruction method of~\cite{Yosef11}.

Despite the plethora of network reconstruction methods,
these have been so far largely limited to explaining a
single experiment or condition. In practice, the network dynamically
changes over time or conditions, calling for reconstructions that
can integrate such data to a coherent picture of the activity
dynamics of the underlying pathways.

Here we tackle this multiple-condition scenario,
where the reconstructed subnetwork should explain in a coherent 
manner multiple experiments driven by the same set of proteins 
(referred to here as {\em anchor} proteins) 
while producing different subsets of affected proteins, or {\em terminals}. 
As in the single-condition case, a parsimonious assumption implies that
the reconstructed subnetwork should be of minimum size. In addition,
we require that its pathways, leading from the anchor to each of the 
terminals, are as homogeneous as possible in terms of the conditions, 
or {\em labels} they span.
We formulate the resulting minimum labeling problem, show that it is
NP complete and characterize its solutions.
We then offer an equivalent formulation that allows us to design
a polynomial integer linear programming (ILP) formulation
for its solution. We implement the ILP algorithm, {\em MKL},
and apply it to two data sets in humans concerning the response to
influenza infection and ER export regulation.
We show that the MKL networks are significantly enriched with respect to the related biological processes and allow obtaining of novel insights on the modeled processes. We further
compare MKL with an extant method, ANAT~\cite{Yosef11}, 
demonstrating the power of our algorithm in integrating data from multiple conditions in a compact and informative manner. 
For lack of space, some algorithmic details are omitted
or deferred to an Appendix.

\section{Preliminaries} \label{problem:description}
Let $G=(V,E)$ be a directed graph, representing a protein-protein
interaction (PPI) network, with vertex set $V$ and edge set $E$.
For a node $v \in V$, denote by $In(v)$ ($Out(v)$) the set of incoming 
(outgoing) edges of $v$, respectively.
Let $L=\{1,\ldots,k\}$ be a set of labels, representing $k\geq 1$ conditions.
Let $f:E\rightarrow 2^{(L)}$ be a labeling function that
assigns each edge of $E$ a possibly empty subset of labels. 
For $1\leq i\leq k$, we define $E_i(f):=\{e\in E:i\in f(e)\}$ to be 
the set of edges with label $i$.
We further denote
$f_{in}(v) = \bigcup_{e \in In(v)} f(e)$ and  
$f_{out}(v) = \bigcup_{e \in Out(v)} f(e)$.

We say that a labeling function $f$ is {\em valid} 
if for every terminal $t$ and each condition $i$ in which it is affected,
there is a path from the anchor to the terminal whose edges are
assigned with the label $i$.
Formally, we require a path from $a$ to $t$ that is restricted to $E_i(f)$.
We evaluate the {\em cost} of the labeling 
according to the number of labels $L(f)$
used and the number of edges $N(f)$ that are assigned with at least one label.
Formally, $L(f)=\sum_{e\in E}|f(e)|$ and 
$N(f)=|\{e \in E : f(e) \neq \emptyset\}|$. The cost is then defined
as $\alpha \cdot L(f) + (1-\alpha) \cdot N(f)$, where $0 \leq \alpha \leq 1$ 
balances the two terms.

We study the following {\bf minimum $k$-labeling (MKL)} problem on $G$:
The input is an anchor node $a\in V$ and 
$k \geq 1$ sets of terminals $T_1,\ldots,T_k$ in $V\setminus \{a\}$
that implicitly assign to each terminal the subset of conditions, or labels 
in which it is affected.
The objective is to find a valid labeling of the edges of $G$ of minimum cost.

Clearly, any valid labeling induces a subnetwork that can model each of the experiments: this subnetwork is comprised of those edges that are assigned a non-empty subset of labels. 
We note that for $k=1$ we have $L(f)=N(f)$, thus in this case the MKL problem is equivalent to the minimum directed Steiner tree problem.
The parameter $\alpha$ balances between two types of solutions: (1) A subnetwork with minimum number of labels ($\alpha=1$), which is equivalent to the union of independent Steiner trees of each of the experiments. (2) A subnetwork with minimum number of edges ($\alpha=0$), which is simply a Steiner tree that spans the union of all sets of terminals. However, general instances of MKL 
where $\alpha\neq 0,1$ can be solved neither by combining the 
independent Steiner trees of each of the experiments nor by constructing a single Steiner tree over all terminals. This is illustrated 
by the toy examples in Figures~\ref{fig:union_counter_example} 
and~\ref{fig:steiner_counter_example}. Next, we provide a characterization of solutions to the MKL problem. 


\begin{theorem}\label{v1theorem}
Given a solution labeling $f$ to an MKL instance,
let $G_i$ denote the subgraph of $G$ that is induced by the edges in $E_i(f)$.
Then $G_i$ is a directed tree rooted at $a$.
\end{theorem}

\begin{proof}
By definition, there is a directed path in $G_i$ from $a$ to each of the 
terminals in $T_i$. Clearly, any edge directed into $a$ can be removed 
without affecting the constraints of a valid solution. 
Thus, it suffices to show that the underlying undirected
graph of $G_i$ contains no cycles. By minimality of the solution, 
every vertex in $G_i$ is reachable from $a$ or else it can be removed
along with its edges. Suppose to the contrary that
$v_1,\ldots,v_n$ is a cycle in the underlying graph.
Since $a$ cannot be on this cycle and by the above observation, each of the
cycle's vertices is reachable from $a$. W.l.o.g., let
$v_1$ be the farthest from $a$ in $G_i$ among all cycle vertices. Then one can obtain
a smaller solution by removing one of the edges $(v_1,v_2)$, $(v_n,v_1)$
(depending on their orientations), a contradiction.
\end{proof}

As noted earlier, when $k=1$ the MKL problem is equivalent to the minimum directed Steiner tree problem, which is known to be NP-complete~\cite{GJ79}. A simple reduction from this case yields the following result:

\begin{theorem}\label{NP-complete} The MKL problem
is NP-complete for every $k \geq 1$.
\end{theorem}

\section{The MKL algorithm}
As the MKL problem is NP-complete, we aim to design an
integer linear program (ILP) for it, which will allow us to solve it to optimality or near-optimality for moderately-sized instances.
In order to design an efficient ILP, we first 
provide an alternative formulation of the MKL problem, expressed in terms of units of flow per label pushed from the anchor toward the terminals.
To this end, we extend the labeling to assign multi-sets rather than sets.
We denote a multi-set by a pair $M=\langle S,\mu \rangle$, where $S$ 
is a set and $\mu \colon S \rightarrow \mathbb{Z}^+$. We say that $x\in M$
if $x\in S$. We let $|M|$ denote the cardinality of the underlying set $S$.

The union $\uplus$ of two multi-sets $\langle S_1,\mu_1 \rangle$, $\langle S_2,\mu_2 \rangle$ is defined as the pair $\langle S,\mu \rangle$, where $S=S_1 \cup S_2$; for every $x \in S_1 \cap S_2$, $\mu(x)=\mu_1(x)+\mu_2(x)$; for $x \in S_1 \setminus S_2$, $\mu(x)=\mu_1(x)$; and for $x \in S_2 \setminus S_1$, $\mu(x)=\mu_2(x)$. We extend the definitions of $f_{in}(v)$ and $f_{out}(v)$ to
multi-sets using this union operator. 
Finally, for a vertex $v\neq a$
we let $L(v)=\{i\in L:v\in T_i\}$; note that for non-terminal nodes $L(v)=\emptyset$.

The alternative objective formulation is as follows: 
Find a multi-set label assignment $g$ that satisfies the 
following constraints:\\
(i) $g_{out}(a)=\langle L,\mu \rangle$, where $\mu(i)=|T_i|$ for every $i \in L$.\\ (The total amount of flow that goes out from the anchor per label equals the number of terminals that belong to the corresponding experiment). \\
(ii) For every $v\neq a$, $g_{in}(v) = g_{out}(v) \uplus L(v)$.\\
(For each label $i$, the incoming flow of a node $v$ equals its outgoing flow, incremented by 1 if $v$ is a terminal expressed in experiment $i$).\\
(iii) Denote $L(g)=\sum_{e\in E}|g(e)|$, $N(g)=|\{e \in E : g(e) \neq \emptyset\}|$, and let $0 \leq \alpha \leq 1$. Then $\alpha \cdot L(g) + (1-\alpha) \cdot N(g)$ is minimal.

We claim that the two formulations are equivalent. Given a multi-set labeling $g$, it is easy to transform it into a labeling $f$ by taking at each edge the underlying set of labels. One can show that the labeling $f$ is valid, i.e. for each $i$ there are paths in $E_i(f)$ that connect $a$ to each of the terminals in $T_i$.
For the other direction, given a labeling $f$ we
can transform it into a multi-set labeling $g$ by defining the multiplicity of a label $i$ at the edge $(u,v) \in E_i(f)$ as the number of terminals from $T_i$ in the subtree of $G_i$ that is rooted at $v$. It is easy to see that all constraints are satisfied by this transformation.

The above problem formulation can be made stricter by requiring that the set of incoming labels to a terminal is exactly the set of labels associated with the terminal. That is, for every terminal $t$ and $i \in L \setminus L(t)$, we 
require that $i \notin g_{in}(t)$. Our ILP formulation includes these
requirements to reflect the experimental observations, but in practice
the strict and non-strict versions produce very similar results.

\subsection{An ILP formulation}
In order to formulate the problem as an integer program, 
we define three sets of variables:
(i) binary variables of the form $y_e^i$, indicating for every $e \in E$ 
and $i \in L$ whether the edge $e$ is tagged with label $i$;
(ii) integer variables of the form $x_e^i$, indicating for every $e \in E$ and 
$i \in L$ the multiplicity of label $i$ (in the range of $0$ to $|T_i|$); and
(iii) binary variables of the form $z_e$, indicating for every $e \in E$ whether the edge $e$ participates in the subnetwork (carrying any label).
For a vertex $v\in V$, let $b_v^i$ be a binary indicator of whether $i\in L(v)$ or not. Let $\alpha$ be some fixed value in the range $[0,1]$.
The formulation is as follows (omitting the constraints on variable ranges):

\begin{eqnarray*}
\min & \alpha \cdot \sum_{e \in E, i \in L} y_e^i + (1-\alpha) \cdot \sum_{e \in E} z_e& \\
\text{s.t.:}&&\\
&y_e^i\leq x_e^i \leq |T_i| \cdot y_e^i & \forall e \in E ,i \in L\\
&y_e^i\leq z_e & \forall e \in E ,i \in L\\
&\sum_{e \in Out(a)} x_e^i = |T_i| & \forall i \in L\\
&\sum_{e \in In(v)} x_e^i = \sum_{e \in Out(v)} x_e^i + b_v^i
& \forall v \in V \setminus \{a\}, i \in L \\
&\sum_{e\in In(t)} y_e^i = 0 & \forall t \in T, i \notin L(t)
\end{eqnarray*}

\subsection{Implementation details and performance evaluation}
We used the commercial IBM ILOG CPLEX optimizer to solve the above ILP. 
Since solving an ILP is time consuming,
we devised a heuristic method for filtering the input network, aiming to capture those edges that the MKL optimal solution is more likely to use.
Specifically, we focused on (directed) edges that lie on a near shortest path (up to one edge longer than a shortest path) between the anchor and any of the terminals. Further, we accepted approximate solutions which enabled our experiments to end within at most two hours. Restricted by this time frame, we attained solutions deviating by at most 5\% and 7\% from the optimal value for the influenza dataset and the ER export dataset (see Experimental Results Section), respectively.

We tested the robustness of MKL to different choices of $\alpha$ on the two 
datasets we analyzed, observing
that the number of edges and labels varied by at most 8\% and 4\%,
respectively, over a wide range of values ($0.25-0.75$).
Thus, we chose $\alpha=0.5$ for our analyses in the sequel.

We evaluated a solution subnetwork using both network-based and biological measures. The network-based measures included the cost and a {\em homogeneity} score. We defined the homogeneity of a node $v$ as the frequency of the most frequent subset of labels among the $t(v)$ terminals under $v$, divided by $t(v)$; the homogeneity score of the subnetwork was then defined as the average over all nodes that span at least two terminals. 
To quantify the biological significance of the reconstructed subnetworks,
we measured the functional enrichment of their internal nodes (non-input nodes) with respect to validation sets that pertain to the process in question.
In addition, we provide expert analysis of the subnetworks.

We compared the performance of our method to that of the state-of-the-art ANAT reconstruction tool~\cite{Yosef11}, which was shown to outperform many existing tools in anchored reconstruction scenarios. 
For each data set, we applied ANAT (with its default parameters, and without the heuristic filtering) to each condition separately, then unified the results to get an integrated subnetwork. We labeled the solution straightforwardly: an edge $e$ was labeled $i$ if $e$ participated in the subnetwork that was constructed for condition $i$.

\section{Experimental results}
We tested the performance of our algorithm on two human data sets that 
concern the cellular response to PR8 influenza virus and ER export regulation. The two data sets were analyzed in the context of a human PPI network reported in~\cite{Yosef11} which contains 44,738 (bidirectional) interactions over
10,169 proteins. We compared our results to those of a previous tool, ANAT~\cite{Yosef11}, applying it independently to each of the terminal sets and taking the union of the subnetworks as the result.
We describe these applications below.

\subsection{Response to influenza infection}
We used data on the response to viral infection by the H1N1 influenza strain A/PR/8/34 ('PR8') in primary human bronchial epithelial cells~\cite{Shapira09}. The data set contains a collection of 135 virus-human PPIs and gene expression profiles, measured at different time points along the course of the infection. We focused on four time points (the ``conditions'') $t=2,4,6,8$ (i.e. $k=4$ labels), in each time point selecting those genes that were differentially expressed above a cutoff of 0.67~\cite{Shapira09}. We did not include time points earlier than $t=2$ or later than $t=8$, as the former had no or very few differentially expressed genes, while the latter induced an order of magnitude larger gene sets that are presumably associated with secondary responses.

We augmented the human network by the influenza-host PPIs and an auxiliary anchor node (named 'virus') which we connected to the 10 viral proteins. After the filtering, the network contained 1,598 proteins and 8,708 interactions.
The four terminal sets contained 8,19,19 and 49
proteins, respectively, with 77 total in their union, out of which 57 were reachable from the anchor. The resulting MKL subnetwork, which is shown in Figure~\ref{fig:pr8_mkl}, contains 127 edges over 123 nodes (117 human, 5 viral and the anchor node) with 60 internal (non-input) nodes. This is in contrast 
to the much larger ANAT solution on this data set, containing 173 nodes
and 106 internal ones. The subnetworks are quite different in terms of node composition, having 31 internal intersecting nodes. A summary of our network-based measures for the two subnetworks can be found in Figure~\ref{fig:stat_comp}.

\begin{figure}[htb]
\begin{center}
\includegraphics[width=\textwidth,bb=25 515 590 760]{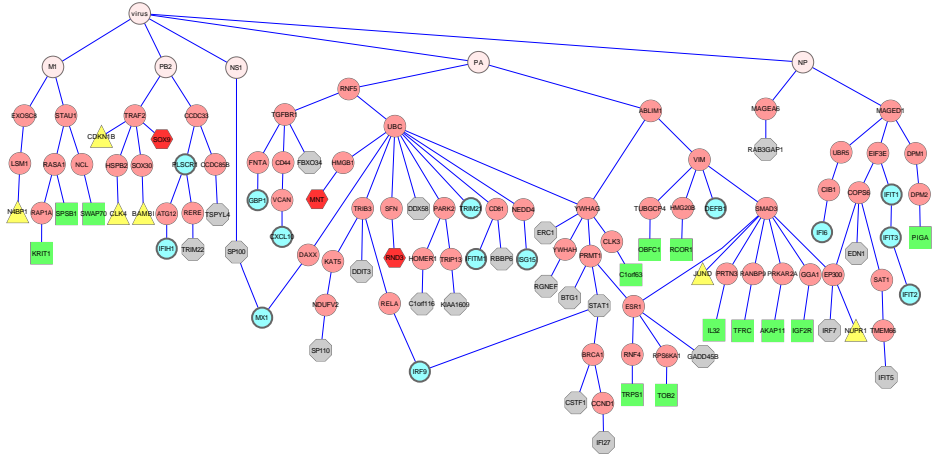}
\caption{\textbf{The MKL subnetwork for the influenza infection data.}
Terminal nodes are marked by their corresponding time point: t=2 - yellow/triangle; t=4 - green/square; t=6 - red/hexagon; t=8 - gray/octagon; more than one time point - cyan oval nodes with thick border. The root is the artificial virus node and the first level is composed solely of viral proteins.}
\label{fig:pr8_mkl}
\end{center}
\end{figure}

\begin{figure}[htb]
\begin{center}
\includegraphics[width=\textwidth,bb=50 640 562 788]{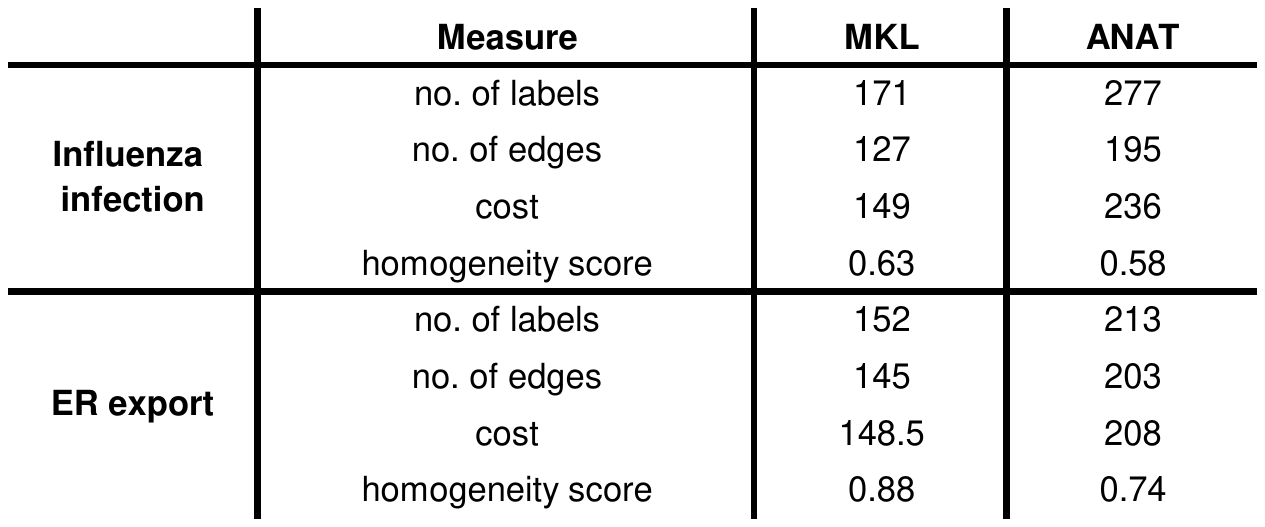}
\caption{Summary statistics for the MKL and ANAT subnetworks on 
the viral infection and ER export data sets.}
\label{fig:stat_comp}
\end{center}
\end{figure}

\begin{figure}[htb]
\begin{center}
\includegraphics[width=\textwidth,bb=50 654 533 788]{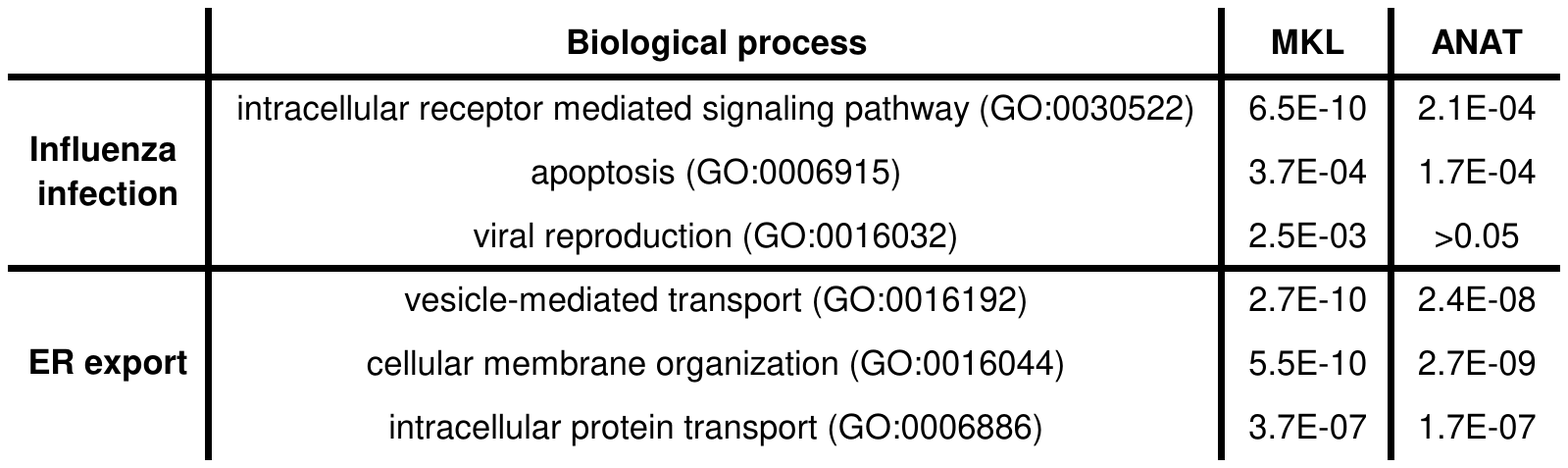}
\caption{Comparison of enrichments of the MKL and ANAT solutions with respect to influenza infection and ER export related processes.}
\label{fig:enrich}
\end{center}
\end{figure}

Next, we scored the enrichment of both subnetworks 
with viral infection related processes such as: viral reproduction, intracellular receptor mediated signaling pathway and apoptosis.
The MKL subnetwork was highly enriched with these processes, outperforming the ANAT subnetwork (Figure~\ref{fig:enrich}).
In the following we present a detailed analysis of the MKL inferred subnetwork and demonstrate its high predictive power and its ability to characterize viral proteins and host mediators in terms of their temporal effect on their targets. Specifically, we show that this subnetwork suggests that an imbalance in the timing of effect between viral proteins (e.g. M1 and NP) or between host mediators (such as Smad3 and UBC) can reveal their different kinetics of influence on host proteins. This is in large contrast to the results produced by the ANAT tool, which does not provide any timing imbalance among downstream targets of viral proteins or host mediators (data not shown due to space 
constraints).

We first present an example of an inferred pathway, selected to demonstrate our MKL approach. The PA-Rnf5-UBC-DAXX-MX1 and NS1-SP100-MX1 paths are a clear example of a predicted pathway that is well supported by extant experimental findings. It is consistent with the known role of both DAXX and SP100 as major components of the PML bodies which control together the localization of MX1 in distinct nuclear components~\cite{Engelhardt04}. Further, DAXX is known to be regulated {\em in vivo} by ubiquitination through UBC and Rnf5~\cite{Wagner11}, supporting our placement of DAXX downstream to UBC.

The MKL network shows that the targets of some human proteins have a common temporal behavior, whereas others have different downstream temporal responses. 
This is consistent with the fact that PPIs naturally represent different mechanisms that might differ in their kinetics. For example, the targets of Traf2 are all early responding genes whereas the targets of Ccdc33 have longer temporal responses. The early effect of Traf2 is consistent with the findings that Traf2 is a signaling transduction kinase protein with fast kinetics. A similar characterization can be applied to other signal transduction proteins such as Smad3. Conversely, the Ccdc33 protein regulates its targets in late time points (6-8 hours) by an unknown mechanism. The results here suggest that this mechanism is orders of magnitude slower than phosphorylation. Similarly, the control of Rnf5 and UBC is expected to show fast kinetics through ubiquitination. In contrast, we find that all the Rnf5/UBC 19 targets are controlled in late time points (6-8 hours), suggesting a novel temporal (late) control on the activity of Rnf5-specific UBC-based ubiquitination during t
 he course of influenza infection.

\subsection{Regulation of endoplasmic reticulum (ER) export}
The journey of secretory proteins, which make up roughly 30\% of the human proteome starts by exit from the ER. Export from the ER is executed by so called COPII vesicles that bud from ER exit sites (ERES).
A protein that is of central importance for ERES biogenesis and maintenance is Sec16A, a large (\textasciitilde 250 kDa) protein that localizes to ERES and interacts with COPII components~\cite{Watson06}.
We have recently performed a siRNA screen to test for kinases and 
phosphatases that regulate the functional organization of the early 
secretory pathway~\cite{Farhan10}. Among the hits identified were 
64 kinases/phosphatases that when depleted result in a reduction in the 
number of ERES. Thus, these are 64 different potential regulators of 
ER export. More recently, a full genome screen tested for genes that 
regulate the arrival of a reporter protein from the ER to the 
cell surface~\cite{Simpson12}. There, the depletion of 45 proteins was shown 
to affect ERES. However, whether the defect in arrival of the reporter 
to the cell surface was due to an effect on ER export or due to 
alterations in other organelles along the secretory route 
(e.g., Golgi apparatus) remains to be determined.

We applied MKL to these two screens, serving as two ``conditions'' highlighting different repertoires of ER export signaling-regulatory pathways.
As the two screens do not intersect (most likely due to differences in read-outs), there were 109 terminals overall, 
85 of them reachable in our human PPI network.
Due to its central importance for ER export and ERES formation, we chose Sec16A as the anchor for this application. After filtering, the network contained 1,907 nodes and 11,329 edges.
The resulting MKL subnetwork, containing 145 nodes and 59 internal ones, 
is depicted in Figure~\ref{fig:erexp_mkl}. 
In comparison, the ANAT solution contains 190 nodes and 104 internal ones (with
35 internal nodes common to the two solutions).
As evident from Figure~\ref{fig:stat_comp}, the MKL solution has substantially lower cost and is more homogeneous.

We assessed the functional enrichment of the MKL subnetwork with biological processes that are of relevance to ER export such as cellular membrane organization, intracellular protein transport and vesicle-mediated transport. All three categories were highly enriched and the $p$-values attained compare favorably to those computed for the ANAT solution (Figure~\ref{fig:enrich}).

Interestingly, 4 proteins of the MKL solution are related to autophagy (two of them internal nodes, $p=0.02$).
Autophagy is an endomembrane-based cellular process that is responsible for capturing and degradation of surplus organelles and proteins. Links between ER export and autophagy have been proposed~\cite{Ishihara01} but there is very limited mechanistic insight into this link. The vesicle-mediated transport process includes the STX17, SNAP29 and ULK1 proteins. The latter is a kinase that initiates the biogenesis of autophagosomes~\cite{Mizushima10}. STX17 and SNAP29 were recently proposed to be involved in autophagy by promoting the formation of ER-mitochondria contact sites and the fusion of autophagosomes with lysosomes~\cite{Hamasaki13,Itakura13}. As the MKL network was generated with terminals and an anchor that regulate ER export, we propose that this approach could be used to identify the molecular link between secretion and autophagy in the future.

\begin{figure}[htb]
\begin{center}
\includegraphics[width=\textwidth,bb=40 500 570 746]{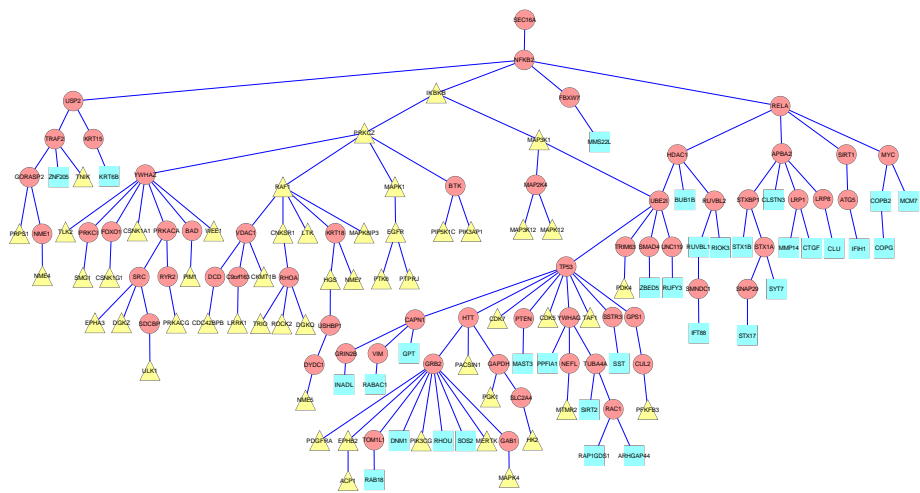}
\caption{\textbf{The MKL subnetwork for the ER export data.}
Terminal nodes are colored/shaped according to the screen they were discovered in: \cite{Farhan10} - yellow/triangle, and \cite{Simpson12} - cyan/square.}
\label{fig:erexp_mkl}
\end{center}
\end{figure}

\section{Conclusions}
The protein-protein interaction network represents a combination of diverse regulation and interaction mechanisms operating in different conditions and
time scales. Integrating such data in a coherent manner to describe a process
of interest is a fundamental challenge, which we aim to tackle in this work
via a novel ILP-based minimum labeling algorithm.
We apply our algorithm to two human data sets and show that it attains 
compact solutions that capture the dynamics of the data and align well
with current knowledge. We expect this type of analysis to gain further 
momentum as composite data sets spanning multiple conditions and time points
continue to accumulate.

\paragraph{Acknowledgments.}
AM was supported in part by a fellowship from the Edmond J. Safra Center for Bioinformatics at Tel Aviv University.
RS was supported by a research grant from the Israel Science Foundation
(grant no. 241/11).

\bibliographystyle{splncs}
\bibliography{labeling}

\begin{thebibliography}{10}

\bibitem{Beisser10}
Beisser, D., Klau, G., Dandekar, T., Mueller, T., Dittrich, M.:
\newblock {BioNet} an {R}-package for the functional analysis of biological
  networks.
\newblock Bioinformatics \textbf{26} (2010)  1129--1130

\bibitem{Huang09}
Huang, S., Fraenkel, E.:
\newblock Integrating proteomic, transcriptional, and interactome data reveals
  hidden components of signaling and regulatory networks.
\newblock Sci. Signal. \textbf{2}(81) (2009)  ra40

\bibitem{Lotem09}
Lotem, E., Riva, L., Su, L., Gitler, A., Cashikar, A., King, O., Auluck, P.,
  Geddie, M., Valastyan, J., Karger, D., Lindquist, S., Fraenkel, E.:
\newblock Bridging high-throughput genetic and transcriptional data reveals
  cellular responses to alpha-synuclein toxicity.
\newblock Nature Genetics \textbf{41} (2009)  316--323

\bibitem{Yosef11}
Yosef, N., Zalckvar, E., Rubinstein, A., Homilius, M., Atias, N., Vardi, L.,
  Berman, I., Zur, H., Kimchi, A., Ruppin, E., Sharan, R.:
\newblock {ANAT}: A tool for constructing and analyzing functional protein
  networks.
\newblock Sci. Signal. \textbf{4} (2011)

\bibitem{GJ79}
Garey, M., Johnson, D.:
\newblock Computers and Intractability: A Guide to the Theory of
  NP-Completeness.
\newblock W.H. Freeman \& Co. (1979)

\bibitem{Shapira09}
Shapira, S., Gat-Viks, I., Shum, B., Dricot, A., Degrace, M., Liguo, W., Gupta,
  P., Hao, T., Silver, S., Root, D., Hill, D., A.Regev, Hacohen, N.:
\newblock A physical and regulatory map of host-influenza interactions reveals
  pathways in {H1N1} infection.
\newblock Cell \textbf{139}(7) (2009)  1255--1267

\bibitem{Engelhardt04}
Engelhardt, O., Sirma, H., Pandolfi, P., Haller, O.:
\newblock Mx1 {GTPase} accumulates in distinct nuclear domains and inhibits
  influenza {A} virus in cells that lack promyelocytic leukaemia protein
  nuclear bodies.
\newblock J Gen Virol. \textbf{85}(8) (2004)  2315--26

\bibitem{Wagner11}
Wagner, S., Beli, P., Weinert, B., Nielsen, M., Cox, J., Mann, M., Choudhary,
  C.:
\newblock A proteome-wide, quantitative survey of in vivo ubiquitylation sites
  reveals widespread regulatory roles.
\newblock Mol Cell Proteomics \textbf{10}(10) (2011)

\bibitem{Watson06}
Watson, P., Townley, A., Koka, P., Palmer, K., Stephens, D.:
\newblock Sec16 defines endoplasmic reticulum exit sites and is required for
  secretory cargo export in mammalian cells.
\newblock Traffic \textbf{7}(12) (2006)  1678--87

\bibitem{Farhan10}
Farhan, H., Wendeler, M., Mitrovic, S., Fava, E., Silberberg, Y., Sharan, R.,
  Zerial, M., Hauri, H.:
\newblock {MAPK} signaling to the early secretory pathway revealed by
  kinase/phosphatase functional screening.
\newblock J Cell Biol. \textbf{189} (2010)  997--1011

\bibitem{Simpson12}
Simpson, J., Joggerst, B., Laketa, V., Verissimo, F., Cetin, C., Erfle, H.,
  Bexiga, M., Singan, V., H\'{e}rich\'{e}, J., Neumann, B., Mateos, A., Blake,
  J., Bechtel, S., Benes, V., Wiemann, S., Ellenberg, J., Pepperkok, R.:
\newblock Genome-wide {RNAi} screening identifies human proteins with a
  regulatory function in the early secretory pathway.
\newblock Nat Cell Biol. \textbf{14}(7) (2012)  764--774

\bibitem{Ishihara01}
Ishihara, N., Hamasaki, M., Yokota, S., Suzuki, K., Kamada, Y., Kihara, A.,
  Yoshimori, T., Noda, T., Ohsumi, Y.:
\newblock Autophagosome requires specific early {Sec} proteins for its
  formation and {NSF/SNARE} for vacuolar fusion.
\newblock Mol Biol Cell. \textbf{12}(11) (2001)  3690--702

\bibitem{Mizushima10}
Mizushima, N.:
\newblock The role of the {Atg1/ULK1} complex in autophagy regulation.
\newblock Curr Opin Cell Biol. \textbf{22}(2) (2010)  132--9

\bibitem{Hamasaki13}
Hamasaki, M., Furuta, N., Matsuda, A., Nezu, A., Yamamoto, A., Fujita, N.,
  Oomori, H., Noda, T., Haraguchi, T., Hiraoka, Y., Amano, A., Yoshimori, T.:
\newblock Autophagosomes form at {ER}-mitochondria contact sites.
\newblock Nature \textbf{495}(7441) (2013)  389--93

\bibitem{Itakura13}
Itakura, E., Mizushima, N.:
\newblock Syntaxin 17: The autophagosomal {SNARE}.
\newblock Autophagy \textbf{9}(6) (2013)

\end{thebibliography}

\appendix

\section{Supplementary Figures}

\renewcommand{\thefigure}{S\arabic{figure}}

\begin{figure}[htb]
\begin{center}
\includegraphics[width=\textwidth,bb=80 520 550 752]{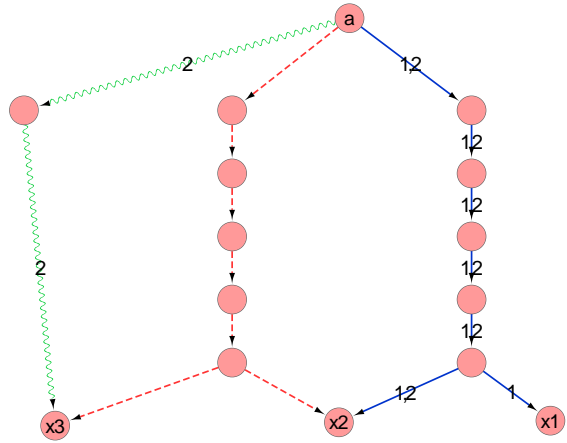}
\end{center}
\caption{\textbf{The optimal MKL solution for $\alpha=0.5$ is neither the union of label-specific Steiner trees nor a subgraph of it.}
In this instance $k=2$, $T_1=\{x_1,x_2\}$ and $T_2=\{x_2,x_3\}$.
The optimal Steiner trees for $T_1$ and $T_2$ are composed of the blue (solid) and red (dashed) edges, resp.
The best MKL solution that uses only edges of the union can be achieved by pushing label 1 over the blue edges and 2 over the red edges, resulting in $14$ labels and $14$ edges.
In contrast, the optimal solution, whose labels appear on top of the figure, contains the blue and green (waved) edges, 
spanning 15 labels and 9 edges.
}
\label{fig:union_counter_example}
\end{figure}

\begin{figure}[h]
\begin{center}
\includegraphics[width=\textwidth,bb=97 565 500 720]{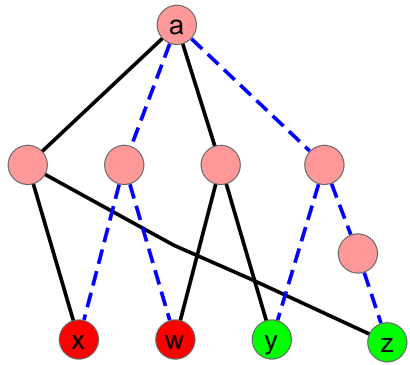}
\end{center}
\caption{\textbf{The optimal MKL solution for $\alpha=0.6$ is not a minimum Steiner tree over all terminals.}
In this instance $k=2$, $T_1=\{x,w\}$ and $T_2=\{y,z\}$. The black (solid) edges form a Steiner tree with 6 edges and 8 labels, whereas the blue (dashed) edges constitute an MKL solution with 7 edges and 7 labels.}
\label{fig:steiner_counter_example}
\end{figure}

\end{document}